\documentclass{llncs}

\usepackage{microtype}

\usepackage{amssymb,amsmath,amsfonts,graphicx,psfig}
\usepackage{times}
\usepackage{xspace}
\usepackage{algorithm}
\usepackage[noend]{algorithmic}
\usepackage{color}
\usepackage{xcolor}

\newcommand{\opt}{\ensuremath{\operatorname{\textsc{Opt}}}\xspace}
\newcommand{\optp}{\ensuremath{\operatorname{P}}\xspace}
\newcommand{\pdd}{\ensuremath{\operatorname{Q}}\xspace}
\newcommand{\optsize}{\ensuremath{\operatorname{\textsc{Opt}}(\sigma)}\xspace}

\newcommand{\A}{\ensuremath{\operatorname{\mathbb{A}}}\xspace}
\newcommand{\B}{\ensuremath{\operatorname{\mathbb{B}}}\xspace}
\newcommand{\alg}{\ensuremath{\operatorname{\mathbb{A}}}\xspace}
\newcommand{\FF}{\ensuremath{\operatorname{\textsc{Ff}}}\xspace}
\newcommand{\NF}{\ensuremath{\operatorname{\textsc{Nf}}}\xspace}
\newcommand{\BF}{\ensuremath{\operatorname{\textsc{Bf}}}\xspace}
\newcommand{\HA}{\ensuremath{\operatorname{\textsc{Ha}}}\xspace}
\newcommand{\BSGA}{\ensuremath{\operatorname{\textsc{Bsga}}}\xspace}
\newcommand{\BSA}{\ensuremath{\operatorname{\textsc{Bsa}}}\xspace}

\newcommand{\SET}[1]{\{#1\}}

\newcommand{\WEHAVE}{\!:\;}
\newcommand{\REALS}{\mathbb{R}}

\newcommand{\cmnt}[1]{}

\def\oh#1{O \left(#1\right)}

\begin{document}

\title{Online Bin Packing with Advice\thanks{The work of the first and third author was partially supported by the Danish Council for Independent Research, Natural Sciences and the Villum Foundation, and most of the work was carried out while these authors were visiting the University of Waterloo.}}

%



\author{Joan Boyar\inst{1},
Shahin Kamali\inst{2},
Kim S. Larsen\inst{1},
Alejandro L\'opez-Ortiz\inst{2}}
\date{}
\pagestyle{plain}

\institute{
University of Southern Denmark, Denmark
\and
University of Waterloo, Canada.
}
\maketitle


\begin{abstract}
We consider the online bin packing problem under the advice complexity model where the ``online constraint'' is relaxed and an algorithm receives partial information about the future requests. We provide tight upper and lower bounds for the amount of advice an algorithm needs to achieve an optimal packing.
We also introduce an algorithm that, when provided with $\log n + o(\log n)$ bits of advice, achieves a competitive ratio of $3/2$ for the general problem. This algorithm is simple and is expected to find real-world applications. We introduce another algorithm that receives $2n + o(n)$ bits of advice and achieves a competitive ratio of $4/3 + \varepsilon$. Finally, we provide a lower bound argument that implies that advice of linear size is required for an algorithm to achieve a competitive ratio better than $9/8$.
\end{abstract}


\section{Introduction}
In the classical one-dimensional bin packing problem the goal is to pack a given sequence of \textit{items} into a minimum number of \textit{bins} with fixed and equal capacities. For convenience, it is assumed that items sizes are in the range $(0,1]$ and the capacities of bins are $1$. In the \textit{online} version of the problem, the items are revealed one by one, and an algorithm must pack each item without any knowledge about future items. The decisions of an online algorithm are irrevocable, i.e., it is not possible to move an item from one bin to another after it is \textit{packed} in a bin.

The online bin packing problem has many applications in practice, from loading trucks subject to weight limitations to creating file backups in removable media~\cite{CoGaJo97}.
Heuristics that have been proposed for the problem include Next-Fit (\NF), First-Fit (\FF), Best-Fit (\BF), and the Harmonic-based class of algorithms. \NF maintains a single \emph{open} bin and places an item in that bin; in the case the item does not fit, it \emph{closes} the bin and opens a new one. \FF keeps a list of bins in the order they are opened, packs an item in the first bin that has enough space, and opens a new bin if necessary. \BF performs similarly to \FF, except that the bins are ordered in increasing order of their remaining capacity. 
Harmonic-based algorithms are based on the idea of packing items of similar sizes together in a bin. For Harmonic$_K$, an item has type $i$ ($1\leq i\leq K-1$) if it is in the range $(\frac{1}{i+1},\frac{1}{i}]$, and type $K$ if it is in the range $(0,\frac{1}{K}]$.
The algorithm applies the \NF strategy for items of each type separately.

As for other online problems, the standard method for comparing bin packing algorithms is competitive analysis. Under competitive analysis, the performance of an algorithm \alg is compared to that of \opt, which is the optimal offline algorithm. More precisely, the competitive ratio of an algorithm \alg is the asymptotically maximum ratio of the cost of \alg to that of \opt for serving the same sequence $\sigma$. \FF and \BF have the same competitive ratio of $1.7$, while the best Harmonic-based algorithm has a competitive ratio of at most $1.58889$~\cite{Seid02}. It is also known that no online algorithm can have a competitive ratio better than $1.54037$~\cite{BalBek12}.

The total lack of information about the future is unrealistic in many real-world scenarios~\cite{EmekFraKorRos2011}. A natural approach for addressing this issue is to relax the problem by providing extra information about the input sequence. For the online bin packing problem, such relaxations have been studied in the contexts of \textit{lookahead}, in which the online algorithm can look at the items arriving in the near future~\cite{Grove95}, and \textit{closed bin packing}, in which the length of the request sequence is known to the online algorithm~\cite{AsgAye02}. In both cases, the average performance of the online algorithm improves, compared to the online algorithms with no information about the future.

The advice complexity model for online algorithms is a more general framework under which the ``no knowledge assumption'' behind online algorithms is relaxed, and the algorithm receives some bits of \textit{advice} about the future requests. The advice can be any information about the input sequence and is generated by an offline oracle which has unbounded computational power. Provided with the appropriate advice, the online algorithms are expected to achieve improved competitive ratios.
The advice model has received significant attention since its introduction~\cite{BocKomKra09,HromKraKra10,EmekFraKorRos2011,BocKomKra11,KommKra11,RenaRosen11,BocKomKra12,BiaBoc12,DebKraMak12,ForKelSte12,KomKraMom12,BocHroKom13,BianBoc13,SieSprUng13}.

In this paper, we study the advice complexity of the online bin packing problem. Our interest in studying the problem under this setting is mostly theoretical. Nevertheless, in many practical scenarios, it can be justified to allow a fast offline oracle to take a ``quick look'' at the input sequence and send some advice to the online algorithm. For example, it may be possible to take a quick look and count the number of items which are larger than $1/2$ and smaller than $2/3$ of the bin capacity. We show that this form of advice can be used to achieve an algorithm which outperforms all online algorithms.


\subsection{Model}
\label{the-model}
In the last few years, slightly different models of advice complexity have been proposed for online problems. All these models assume that there is an offline oracle with infinite computational power, which provides the online algorithm with some bits of advice. How these bits of advice are given to the algorithm is the source of difference between the models. In the first model, presented in~\cite{DobKraPar09}, an online algorithm poses a series of questions which are answered by the offline oracle in \textit{blocks of answers}. The total size of the answers, measured in the number of bits, defines the advice complexity. The problem with this model is that a lot of information can be encoded in the individual length of each block. To address this issue, another model is proposed in~\cite{EmekFraKorRos2011} which assumes that online algorithms receive a fixed number of bits of advice per request. We call this model the \textit{advice-with-request model}. This model is studied for problems, such as metrical task systems and $k$-server, and the results tend to use at least a constant number of bits of advice per request~\cite{EmekFraKorRos2011,RenaRosen11}. \cmnt{under advice-per-request model, the advice is mostly constant per request}
Nevertheless, there are many online problems for which a sublinear and even a constant number of bits of advice in total is sufficient to achieve good competitive ratios. However, under the advice-with-request model, the possibility of sending a sublinear number of advice bits to the algorithm is not well defined. In~\cite{BocKomKra09,BocKomKra11} another model of advice complexity is presented which assumes that the online algorithm has access to an \textit{advice tape}, written by the offline oracle. At any time step, the algorithm may refer to the tape and read any number of advice bits. The advice complexity is the number of bits on the tape accessed by the algorithm. We refer to this model as \textit{advice-on-tape model}. Since its introduction, the advice-on-tape model has been used to analyze the advice complexity of many online problems including paging~\cite{BocKomKra09,HromKraKra10,KommKra11}, disjoint path allocation~\cite{BocKomKra09}, job shop scheduling~\cite{BocKomKra09,KommKra11}, $k$-server~\cite{BocKomKra11,RenaRosen11}, knapsack~\cite{BocKomKra12},
various coloring problems~\cite{BiaBoc12,ForKelSte12,BianBoc13,SieSprUng13},
set cover~\cite{KomKraMom12,BocHroKom13}, maximum clique~\cite{BocHroKom13}, and graph exploration~\cite{DebKraMak12}.

Under the advice-on-tape model, we require a mechanism to infer how many bits of advice the algorithm should read at each time step. This could be implicitly derived during the execution of the algorithm or explicitly encoded in the advice string itself. For example, we may use a \textit{self-delimited} encoding as used in \cite{BocKomKra11}, in which the value of a non-negative integer $X$ is encoded by writing the value of $\lceil \log (\lceil \log (X+1) \rceil+1) \rceil$ in unary (a string of 1's followed by a zero), the value of $\lceil \log (X+1) \rceil$ in binary \footnote{In this paper we use $\log n$ to denote $\log_2(n)$.}, and the value of $X$ in binary. These codes respectively require $\lceil\log (\lceil\log (X+1) \rceil+1) \rceil +1$, $\lceil\log (\lceil\log (X+1) \rceil+1) \rceil$, and 
$\lceil \log (X+1) \rceil$ bits. Thus, the self-delimited encoding of $X$ requires
\[ \hspace{1.5cm} e(X) = \lceil \log (X+1) \rceil+ 2\lceil\log (\lceil\log (X+1) \rceil+1) \rceil + 1 \] 
bits.
The existence of self-delimited encodings at the beginning of the tape usually adds a lower-order term to the number of advice bits required by an algorithm.

Regarding notation, we use $\alg(\sigma)$ to denote the costs of \alg for packing a request sequence $\sigma$. When $\sigma$ follows from the context, we simply use $\alg$ to denote this cost. We use similar notation for all algorithms, including \opt.

We consider the bin packing problem under the advice-on-tape model, which is formally defined as follows, based on the definition of the advice model in~\cite{BocKomKra11}:
\begin{definition}
In the {\em online bin packing problem with advice}, the input is a sequence of items $\sigma = \langle x_1, \ldots, x_n \rangle$, revealed to the algorithm in an online manner $(0<x_i\leq 1)$. The goal is to pack these items in the minimum number of bins of unit size. At time step $t$, an online algorithm should pack item $x_t$ into a bin. The decision of the algorithm to select the target bin is a function of $\Phi, x_1, \ldots, x_{t-1}$, where $\Phi$ is the content of the advice tape. An algorithm \alg is {\em $c$-competitive with advice complexity
$s(n)$} if there exists a constant $c_0$ such that, for all $n$ and for all input sequences $\sigma$ of length at most $n$, there exists some advice $\Phi$ such that $\alg(\sigma) \leq c \ \optsize + c_0$, and at most the first $s(n)$ bits of $\Phi$ have been accessed by the algorithm. If $c=1$ and $c_0 = 0$, then \alg is {\em optimal}.
\end{definition}


\subsection{Contribution}
We answer different questions about the advice complexity of the online bin packing problem. First, we study how many bits of advice are required to achieve an optimal solution. We consider two different settings of the problem. When there is no restriction on the number of distinct items or their sizes, we present the easy result that $n \lceil \log \optsize \rceil$ 
bits of advice are sufficient to achieve an optimal solution, where $\optsize$ is the number of bins in an optimal packing. We also prove that at least $(n - 2 \optsize) \log \optsize$ bits of advice are required to achieve an optimal solution.

When there are $m$ distinct items in the sequence, we prove that at least $(m-3) \log n -2 m \log m$ bits of advice are required to achieve an optimal solution. If $m$ is a constant, there is a linear time online algorithm that receives $m \log n + o(\log n)$ bits of advice and achieves an optimal solution. We also show that, even if $m$ is not a constant, there is a polynomial time online algorithm that receives $m \lceil \log (n+1) \rceil + o(\log n)$ bits of advice and achieves a packing with $(1+ \varepsilon)\optsize + 1$ bins. 

We also study a relevant question that asks how many bits of advice are required to perform strictly better than all online algorithms. We bound this by providing an algorithm which receives $\log n + o(\log n)$ bits of advice and achieves a competitive ratio of $3/2$. Recall that any online bin packing algorithm has a competitive ratio of at least $1.54037$~\cite{BalBek12}. Hence, our algorithm outperforms all online algorithms.  

Moreover, we introduce an algorithm that receives $2n + o(n)$ bits of advice and achieves a competitive ratio of $4/3+\varepsilon$, for any fixed value of $\varepsilon>0$. \cmnt{In our algorithm, we need $\varepsilon < 1/11$}
We also prove a lower bound that implies that a linear number of bits of advice are required to achieve a competitive ratio of $9/8 - \delta$ for any fixed value of $\delta>0$.

%
\section{Optimal Algorithms with Advice}\label{secOpt}

In this section we study the amount of advice required to achieve an optimal solution. We first investigate the theoretical setting in which there is no restriction on the number of distinct items or on their sizes.
We observe that there is a simple algorithm that receives $n \lceil \log \optsize \rceil$ 
bits of advice and achieves an optimal solution. Such an algorithm basically reads $\lceil \log \optsize \rceil$ bits for each item, encoding the index of the bin that includes the item in an optimal packing. We show that the upper bound given by this algorithm is tight up to lower order terms, when $n-2\optsize\in\Theta(n)$.

\begin{theorem}\label{thUpperTheo}
To achieve an optimal packing for a sequence of size $n$ and optimal cost $\optsize$, it is sufficient to receive $n \lceil \log \optsize \rceil$ bits of advice. Moreover, any deterministic online algorithm requires at least $(n-2\optsize) \log \optsize$ bits of advice to achieve an optimal packing.
\end{theorem}

\begin{proof}
\ \\
\textit{Upper Bound:}
Consider an offline oracle that knows an optimal packing (note that such an oracle has unbounded computational power). This oracle simply writes on the advice tape, for each item $x$, except for the last two, the index of the bin in an optimal packing that $x$ is packed in.
To pack any item $x$, the online algorithm simply reads the index of the bin that $x$ should be packed in and packs $x$ accordingly. For the last two items, the algorithm simply uses Best-Fit. Since the packing is the same as one for an optimal algorithm up to that point, if it is impossible to fit both of the remaining items in the bins already used, Best-Fit will ensure that at least one fits if that is possible. If both of the remaining items fit in the same already open bin, it is fine to put the first one of the last two items anywhere it fits, since there will still be space remaining for the last. If both of the remaining items fit in open bins, but should be in different bins, using Best-Fit will ensure that they are both placed there. \\
This requires $\lceil \log \optsize \rceil$ bits of advice per item which sums up to $(n-2) \lceil \log \optsize \rceil$ bits of advice. 
The algorithm should also know the value of $X = \lceil \log \optsize \rceil$ in order to read the appropriate number of bits on each request. This can be done by encoding $X$ in unary and terminating with a zero. This uses no more than $2\lceil \log \optsize \rceil$ bits.
Consequently the number of advice bits used by the algorithm is $n \lceil \log \optsize \rceil$ as stated by the theorem. \ \\ \ \\
\textit{Lower Bound:} 
Consider a set $S = \SET{ \sigma_1, \ldots, \sigma_N }$ of sequences, so that each $\sigma_r$ has length $n$ for $1 \leq r \leq N$. Let $1\leq k\leq n-1$.
Each sequence $\sigma_r$ in the set has the form $$\left\langle \frac{1}{4}, \frac{1}{ 8}, \frac{1}{16}, \ldots ,\frac{1}{2^{n-k+1}}, u^r_1, u^r_2, \ldots ,u^r_{k}\right\rangle$$ in which $u^r_1, \ldots ,u^r_{k}$ are defined from a set $V$ of vectors in form $V_r = (v^r_1=1, v^2_2=2, \ldots, v^r_k=k, v^r_{k+1}, v^r_{k+2}, \ldots, v^r_{n-k})$ such that each $ v^r_h \in \SET{ 1, \ldots, k } $ for $1 \leq h \leq n-k$. 

\begin{quote}
For example, when $n=8$ and $k = 3$, the vector $(1,2,3,2,1)$ is a vector in $V$. 
\end{quote}

We associate with each vector $V_r \in V$ a sequence $\sigma_r \in S$. For a vector $V_r \in V$ and bin $j$, define $ u^r_j = 1 - \sum\limits_{\substack{1 \leq i \leq n-k \\ v_i^r=j}} a_i $, where $a_i$ is the $i$th item in the sequence $\sigma_r$, i.e., $a_i = \frac{1}{2^{i+1}}$. 
Note that all $u_j$s are strictly larger than $0.5$. Clearly, ${\textsc{Opt}}(\sigma_r) = k$ for all $r$.
We refer to the first $n-\optsize$ items as \textit{small} items and the last $\optsize$ items as \textit{large} items.

\begin{quote}
For example, assume $n=8$ and $\optsize = 3$. For a vector $V_r = (1,2,3,2,1)$, we have $u^r_1 = 1 - (\frac{1}{4} + \frac{1}{64}) = 0.734375$, $u^r_2 = 1 - (\frac{1}{8} + \frac{1}{32}) = 0.84375$, and $u^r_3 = 1 - \frac{1}{16} = 0.9375$. Hence, the sequence $\sigma_r$ associated with $V_r$ is $\langle \frac{1}{4}, \frac{1}{8}, \frac{1}{16}, \frac{1}{32}, \frac{1}{64}, 0.734375, 0.84375, 0.9375 \rangle$.
\end{quote}

In fact, $V_r$ indicates in which bin each of the first $n-\optsize$ items of $\sigma_r$ should be packed, and at the end, $u^r_j$ fills the empty space of the $j$th bin to capacity to achieve an optimal packing $P$ for a given sequence (it is optimal since all bins are fully packed). The restriction that the sequence starts with $k$ distinct items  
ensures that we do not need to consider permutations of the bins in $P$ as additional optimal packings. We claim that $P$ is the unique optimal packing. Suppose there is another optimal packing $P'$. Observe that each bin includes at most one large item, and indeed exactly one since we assume it is also optimal.
Let $a_i (1\leq i \leq n-\optsize)$ be the first item which is packed in some other bin in $P'$ than the one prescribed by $P$. Consider the bin $B$ that $a_i$ is packed into in $P$. This bin cannot be fully packed in $P'$ since $a_i$ is strictly larger than the total size of all remaining small items, i.e., even if we put all of them in the empty space of $a_i$, there is still some empty space in $B$. As a result $P'$ cannot be optimal. Hence there is unique solution for packing each sequence in the set $S$.

Note that there are 
 $N = \optsize ^ {n-2\optsize}$ 
sequences $S$. 
We claim that these sequences need separate advice strings. Suppose otherwise, and let $\sigma_r, \sigma_{r'} \in S$ $(r\neq r')$ be two different sequences with the same advice string. Note that the first $n-\optsize$ items in these sequences are the same. Since the online algorithm performs deterministically and we assume it receives the same advice for both $\sigma_{r}$ and $\sigma_{r'}$, the partial packings of the algorithms after serving the first $n-\optsize$ items are the same for both sequences. However, as discussed earlier, this implies that the final packing of the algorithm is different from the optimal packing prescribed by $V_{r''}$ for at least one of the sequences. As discussed, such a packing is the unique optimal packing and deviating from that increases the cost of the algorithm by at least one unit. As a result, the algorithm performs non-optimally for at least one of $\sigma_{r}$ or $\sigma_{r'}$. We conclude that the sequences in the set $S$ need separate advice strings. Since there are 
$N = \optsize ^ {n-2\optsize}$ 
sequences in $S$, at least 
$\log (\optsize ^ {n-2\optsize})= (n-2\optsize) \log \optsize$ 
bits of advice are required to get that many distinct advice strings.
\qed
\end{proof}

Next, we consider a more realistic scenario where there are $m \in o(n)$ distinct items and the values of these items are known to the algorithm. 
Assume that the advice tape specifies the number of items of each size. If we are not concerned about the running time of the online algorithm, there is enough information to obtain an optimal solution. If we are concerned, we can use known results for solving the offline problem ~\cite{BaewEck83,VegaLuk81,Vazira04}. We formalize this in what follows. 

\begin{lemma}[\cite{BaewEck83}] \label{lemVaz}
Consider the restriction of the bin packing problem to instances in which
the number of distinct item sizes is a constant non-negative integer $m$. There is a linear time algorithm that optimally solves this restricted problem.
\end{lemma}

If there are more than a constant number of distinct items sizes, we can solve
the problem almost optimally if the item sizes are lower bounded by a fixed value $\varepsilon$.

\begin{lemma}[\cite{VegaLuk81,Vazira04}] \label{lemVaz2}
There is a polynomial algorithm for the bin packing problem which opens at most $(1+\varepsilon)\optsize + 1$ bin, in which $\varepsilon$ is any small but constant value.
\end{lemma}

We use the above results to otain the following:

\begin{theorem}\label{thUpperPrac}
Consider the online bin packing problem in which there are $m$ distinct items. If $m$ is a constant, there is a (linear time) optimal online algorithm that receives $m \log n + o(\log n)$ bits of advice. If $m$ is not a constant, there is a (polynomial time) online algorithm that reads $m \lceil \log (n+1) \rceil + o(\log n)$ bits of advice and achieves an almost optimal packing with at most $(1+ \varepsilon) \optsize +1$ bins, for any small but constant value of $\varepsilon$.
\end{theorem}

\begin{proof}
The offline oracle simply encodes the input sequence, considered as
a multi-set, in $m \lceil \log (n+1) \rceil$ bits of advice. In order to do that, it writes the number of occurrences of each of the $m$ distinct items on the tape. 
The online algorithm uses the algorithms of Lemma~\ref{lemVaz} (for constant values of $m$) or that of Lemma~\ref{lemVaz2} (for non-constant $m$) to compute an (almost) optimal packing. Then it packs the items in an online manner according to such an (almost) optimal packing.
 The algorithms reads frequencies of items in chunks of $X=\lceil \log (n+1) \rceil$ bits and consequently needs to know the value of $X$.
So, we add self-delimited encodings of $X$ at the beginning of the tape using $e(X)$ bits.
The number of advice bits used by the algorithm is thus $m \lceil \log (n+1) \rceil + \oh{\log\log n}$, which is $m \lceil \log (n+1) \rceil + o(\log n)$ as $m \in o(n)$.
\qed
\end{proof}

We show that the above upper bound is asymptotically tight. We start with the following simple lemma.

\begin{lemma}\label{ntlemma}
Consider the equation $x_1+2x_2+ \ldots + \alpha x_{\alpha} = X$ in which the $x_i$s ($i\leq\alpha)$ and $X$ are non-negative integers. If $X$ is sufficiently large, then this equation has at least $\left( 1 + \frac{2X}{\alpha(\alpha+1)} \right) ^ {\alpha -1}$ solutions.
\end{lemma}

\begin{proof}
Define $A = \sum_{i=1} ^{\alpha} i$. Assign arbitrary values in the range $[0..X/A]$ to all $x_i$s for $2\leq i\leq \alpha$ (for simplicity assume $X/A$ is an integer). There are $(1+X/A)^{\alpha-1}$ different such assignments. Any of these assignments defines a valid solution for the equation since by definition of $A$ we have $\sum_{i=2} ^{\alpha} i x_i \leq X$, and we can assign $x_1 = X -\sum_{i=2} ^{\alpha} i x_i$. Replacing $A$ with $\alpha(\alpha+1)/2$ completes the proof.
\qed
\end{proof}

\begin{theorem}
 At least $(m-3) \log n -2 m \log m$ bits of advice are required to achieve an optimal solution for the online bin packing problem on sequences of length $n$ with $m$ distinct items, each of size at least $\frac{1}{2m}$.
\end{theorem}

\begin{proof}
We define a family of sequences of length $n$ and containing $m$ distinct items and show that the sequences in this family need separate advice strings to be optimally served by an online algorithm. To define the family, we fix $m$ item sizes as being $\SET{ \frac{1}{2m}, \frac{m+2}{2m}, \frac{m+3}{2m}, \ldots ,\frac{2m-1}{2m}, 1 }$. To simplify the argument, we scale up the sizes of bins and items by a factor of $2m$. So, we assume the item sizes are $\SET{1 , m+2, m+3 \ldots , 2m-1, 2m}$, and the bins have capacity $2m$. Each sequence in the family starts with $n/2$ items of size $1$. Consider any packing of these items in which all bins have level at most equal to $m-2$. Such a packing includes $a_1$ bins of level~1 (one item of size 1 in each), $a_2$ bins of level~2 (two items of size 1 in each), etc., such that the $a_i$s are non-negative integers and $a_1 + 2a_2 + \ldots + (m-2) a_{m-2} = n/2$. By Lemma~\ref{ntlemma}, there are at least $\left( 1+ \frac{n}{(m-1)(m-2)}\right)^{m-3}$ distinct packings with the desired property. For any of these packings, we define a sequence in our family. Such a sequence starts with $n/2$ items of size $1$ and is followed by another $n/2$ items. Let $B$ denote the number of bins in a given packing of the first $n/2$ items, so that $B \leq n/2$. The sequence associated with the packing is followed by $B$ items of size larger than $m+1$ which completely \textit{fit} these bins (in non-increasing order of their sizes). Finally, we include another $n/2-B$ items of size $2m$ in the sequence to achieve a sequence of length $n$.

We claim that any of the sequences in the family has a unique optimal packing of size $n/2$. This is because there are exactly $n/2$ \textit{large} items of size strictly greater than $m$ (more than half the capacity of the bin), and the other $n/2$ items have \textit{small} size $1$ (which fit the empty space of all bins). So each bin is fully packed with one large item of size $x$ and $2m-x$ items of size $1$ (see Figure~\ref{packs}).

The unique optimal packing of each sequence is defined by the partial packing of the first $n/2$ small items. Consider a deterministic online algorithm \A receiving the same advice string for two sequence $\sigma_1$ and $\sigma_2$. Since \A is deterministic and both sequences start with the same sub-sequence of small items, the partial packing of the algorithm after packing the first $n/2$ items is the same for both $\sigma_1$ and $\sigma_2$. As a result, the final packing of \A is sub-optimal for at least one them. We conclude that any deterministic online algorithm should receive distinct advice strings for each sequence in the family. Since there are at least $\left(1+\frac{n}{(m-1)(m-2)}\right)^{m-3}$ sequences in the family, at least $(m-3) \log \left(1+\frac{n}{(m-1)(m-2)}\right) > (m-3) \log n -2 m \log m$ bits of advice are required. 
\qed
\end{proof}

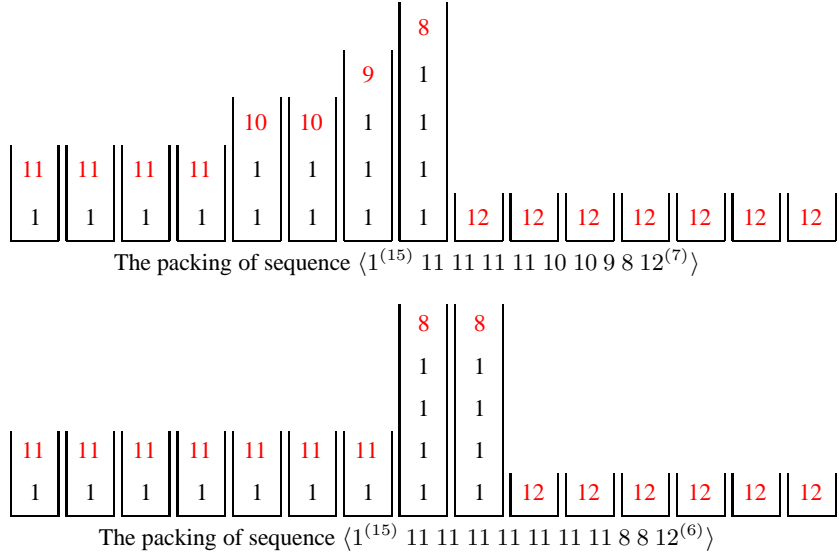
\begin{figure} [!t]
\centering
\begin{picture}(330,90)
\color{black}
\put(15,0){\line(0,36){36}}
\put(33,0){\line(0,36){36}}
\put(15,0){\line(18,0){18}}
\put(22,6){{1}}
\color{red}
\put(19,24){{11}}
\color{black}
\color{black}
\put(36,0){\line(0,36){36}}
\put(54,0){\line(0,36){36}}
\put(36,0){\line(18,0){18}}
\put(43,6){{1}}
\color{red}
\put(40,24){{11}}
\color{black}
\color{black}
\put(57,0){\line(0,36){36}}
\put(75,0){\line(0,36){36}}
\put(57,0){\line(18,0){18}}
\put(64,6){{1}}
\color{red}
\put(61,24){{11}}
\color{black}
\color{black}
\put(78,0){\line(0,36){36}}
\put(96,0){\line(0,36){36}}
\put(78,0){\line(18,0){18}}
\put(85,6){{1}}
\color{red}
\put(82,24){{11}}
\color{black}
\color{black}
\put(99,0){\line(0,54){54}}
\put(117,0){\line(0,54){54}}
\put(99,0){\line(18,0){18}}
\put(106,6){{1}}
\put(106,24){{1}}
\color{red}
\put(103,42){{10}}
\color{black}
\color{black}
\put(120,0){\line(0,54){54}}
\put(138,0){\line(0,54){54}}
\put(120,0){\line(18,0){18}}
\put(127,6){{1}}
\put(127,24){{1}}
\color{red}
\put(124,42){{10}}
\color{black}
\color{black}
\put(141,0){\line(0,72){72}}
\put(159,0){\line(0,72){72}}
\put(141,0){\line(18,0){18}}
\put(148,6){{1}}
\put(148,24){{1}}
\put(148,42){{1}}
\color{red}
\put(148,60){{9}}
\color{black}
\color{black}
\put(162,0){\line(0,90){90}}
\put(180,0){\line(0,90){90}}
\put(162,0){\line(18,0){18}}
\put(169,6){{1}}
\put(169,24){{1}}
\put(169,42){{1}}
\put(169,60){{1}}
\color{red}
\put(169,78){{8}}
\color{black}
\color{black}
\put(183,0){\line(0,18){18}}
\put(201,0){\line(0,18){18}}
\put(183,0){\line(18,0){18}}
\color{red}
\put(187,6){{12}}
\color{black}
\color{black}
\put(204,0){\line(0,18){18}}
\put(222,0){\line(0,18){18}}
\put(204,0){\line(18,0){18}}
\color{red}
\put(208,6){{12}}
\color{black}
\color{black}
\put(225,0){\line(0,18){18}}
\put(243,0){\line(0,18){18}}
\put(225,0){\line(18,0){18}}
\color{red}
\put(229,6){{12}}
\color{black}
\color{black}
\put(246,0){\line(0,18){18}}
\put(264,0){\line(0,18){18}}
\put(246,0){\line(18,0){18}}
\color{red}
\put(250,6){{12}}
\color{black}
\color{black}
\put(267,0){\line(0,18){18}}
\put(285,0){\line(0,18){18}}
\put(267,0){\line(18,0){18}}
\color{red}
\put(271,6){{12}}
\color{black}
\color{black}
\put(288,0){\line(0,18){18}}
\put(306,0){\line(0,18){18}}
\put(288,0){\line(18,0){18}}
\color{red}
\put(292,6){{12}}
\color{black}
\color{black}
\put(309,0){\line(0,18){18}}
\put(327,0){\line(0,18){18}}
\put(309,0){\line(18,0){18}}
\color{red}
\put(313,6){{12}}
\color{black}
\end{picture} \\
\small{The packing of sequence $\langle 1^{(15)} \ 11 \ 11 \ 11 \ 11 \ 10 \ 10\ 9 \ 8 \ 12^{(7)} \rangle$}
\begin{picture}(330,90)
\color{black}
\put(15,0){\line(0,32){32}}
\put(33,0){\line(0,32){32}}
\put(15,0){\line(18,0){18}}
\put(22,6){{1}}
\color{red}
\put(19,22){{11}}
\color{black}
\color{black}
\put(36,0){\line(0,32){32}}
\put(54,0){\line(0,32){32}}
\put(36,0){\line(18,0){18}}
\put(43,6){{1}}
\color{red}
\put(40,22){{11}}
\color{black}
\color{black}
\put(57,0){\line(0,32){32}}
\put(75,0){\line(0,32){32}}
\put(57,0){\line(18,0){18}}
\put(64,6){{1}}
\color{red}
\put(61,22){{11}}
\color{black}
\color{black}
\put(78,0){\line(0,32){32}}
\put(96,0){\line(0,32){32}}
\put(78,0){\line(18,0){18}}
\put(85,6){{1}}
\color{red}
\put(82,22){{11}}
\color{black}
\color{black}
\put(99,0){\line(0,32){32}}
\put(117,0){\line(0,32){32}}
\put(99,0){\line(18,0){18}}
\put(106,6){{1}}
\color{red}
\put(103,22){{11}}
\color{black}
\color{black}
\put(120,0){\line(0,32){32}}
\put(138,0){\line(0,32){32}}
\put(120,0){\line(18,0){18}}
\put(127,6){{1}}
\color{red}
\put(124,22){{11}}
\color{black}
\color{black}
\put(141,0){\line(0,32){32}}
\put(159,0){\line(0,32){32}}
\put(141,0){\line(18,0){18}}
\put(148,6){{1}}
\color{red}
\put(145,22){{11}}
\color{black}
\color{black}
\put(162,0){\line(0,80){80}}
\put(180,0){\line(0,80){80}}
\put(162,0){\line(18,0){18}}
\put(169,6){{1}}
\put(169,22){{1}}
\put(169,38){{1}}
\put(169,54){{1}}
\color{red}
\put(169,70){{8}}
\color{black}
\color{black}
\put(183,0){\line(0,80){80}}
\put(201,0){\line(0,80){80}}
\put(183,0){\line(18,0){18}}
\put(190,6){{1}}
\put(190,22){{1}}
\put(190,38){{1}}
\put(190,54){{1}}
\color{red}
\put(190,70){{8}}
\color{black}
\color{black}
\put(204,0){\line(0,16){16}}
\put(222,0){\line(0,16){16}}
\put(204,0){\line(18,0){18}}
\color{red}
\put(208,6){{12}}
\color{black}
\color{black}
\put(225,0){\line(0,16){16}}
\put(243,0){\line(0,16){16}}
\put(225,0){\line(18,0){18}}
\color{red}
\put(229,6){{12}}
\color{black}
\color{black}
\put(246,0){\line(0,16){16}}
\put(264,0){\line(0,16){16}}
\put(246,0){\line(18,0){18}}
\color{red}
\put(250,6){{12}}
\color{black}
\color{black}
\put(267,0){\line(0,16){16}}
\put(285,0){\line(0,16){16}}
\put(267,0){\line(18,0){18}}
\color{red}
\put(271,6){{12}}
\color{black}
\color{black}
\put(288,0){\line(0,16){16}}
\put(306,0){\line(0,16){16}}
\put(288,0){\line(18,0){18}}
\color{red}
\put(292,6){{12}}
\color{black}
\color{black}
\put(309,0){\line(0,16){16}}
\put(327,0){\line(0,16){16}}
\put(309,0){\line(18,0){18}}
\color{red}
\put(313,6){{12}}
\color{black}
\end{picture} \\
\small{The packing of sequence $\langle 1^{(15)} \ 11 \ 11 \  11 \ 11 \ 11 \ 11 \ 11 \ 8 \ 8 \ 12^{(6)} \rangle$} 
\caption{The optimal packings for two sequences of the family when $n=30$ and $m=6$ (item sizes and bin capacities are scaled by $2m=12$).}
\label{packs}%
\end{figure}

\section{An Algorithm with Sublinear Advice} \label{secsublinear}
In what follows we introduce an algorithm that receives $\log n + o(\log n)$ bits of advice and achieves a competitive ratio of $\frac{3}{2}$, for any instance of the online bin packing problem.
An offline oracle can compute and write the advice on the tape in linear time, and
the online algorithm runs as fast as First-Fit. Thus, the algorithm might be applied in practical scenarios in which it is allowed to have a ``quick look'' at the input sequence. 

We call items \textit{tiny}, \textit{small}, \textit{medium}, and \textit{large} if their sizes lie in the intervals $(0,1/3]$, $(1/3,1/2]$, $(1/2,2/3]$, and $(2/3,1]$, respectively. The advice that the algorithm receives is the number of medium items, which we denote by $\alpha$.

The algorithm reads the advice tape, obtains $\alpha$, opens $\alpha$ bins, called \textit{critical bins}, and reserves $2/3$ of the space in each of them. This reserved space will be used to pack a medium item in each of the critical bins, and these bins have a \textit{virtual level} of size 2/3 at the beginning.
All other bins have virtual level zero when they are opened.
 The algorithm serves an item $x$ in the following manner:

\begin{itemize}
\item If $x$ is a large item, open a new bin for it. Set the virtual level to its size.
\item If $x$ is a medium item, put it in the reserved space of a critical bin $B$. Update the virtual level to the actual level. ($B$ will not have any reserved space now.)
\item If $x$ is small or tiny, use the First Fit (\FF) strategy to put it into any of the open bins, based on virtual levels (open a new bin if required). Add the size of the item to the virtual level.
\end{itemize}

Note that the critical bins appear first in the ordering maintained by the algorithm as they are opened before other bins. 

\begin{theorem}\label{sublinAdvice}
There is an online algorithm which receives $\log n + o(\log n)$ bits of advice and has cost $3/2 \opt(\sigma) + 3$ for serving any sequence $\sigma$ of size $n$.
\end{theorem}

\begin{proof}
We prove that the algorithm described above has the desired property. The value of $\alpha$ is encoded in $X = \lceil\log (n+1) \rceil$ bits of advice. In order to read this properly from the tape, the algorithm needs to know the value of $X$. This can be done by adding the self-delimited encoding of $X$ in $e(X) =
\lceil \log X \rceil + 2\lceil \log \log (X) \rceil+2$
bits at the beginning of the tape. Consequently the number of advice bits used by the algorithm is $X + \oh{\log X}$, which is $\log n + o(\log n)$ as stated by the theorem.

Consider the final packing of the algorithm for serving a sequence $\sigma$.
There are two cases. In the first case, there is a critical bin $B$ so that no other item, except a medium item, is packed in it. Since all tiny items are smaller than $1/3$ and can fit in $B$, all the non-critical bins that are opened after $B$ include small and large items only. More precisely, they include either a single large item or two small items (except the last one which might have a single small item). Let $L$, $M$, and $S$ denote the number of large, medium, and small items. The cost of the algorithm is at most $L+M+S/2+1$. Now, if $S \leq M$, this would be at most $L +3/2 M+1$. Since $L+M$ is a lower bound on the cost of \opt,
the cost of the algorithm is at most $3/2 \opt(\sigma)+1$ and we are done. If $S>M$, \opt should open $L+M$ bins for large and medium items, and in the best case, it packs $M$ small items together with medium ones. For the other $S-M$ bins, \opt has to open at least $(S-M)/2$ bins. Hence the cost of \opt is at least $L+M + (S-M)/2 = L + M/2 + S/2$, and we have $3/2 \opt(\sigma) \geq 3L/2 + 3M/4 + 3S/4 > L + M + S/2$. Thus, the cost of the algorithm is at most $3/2 \opt(\sigma)+1$.

In the second case, we assume that all critical bins include another item in addition to the medium item. We claim that at the end of serving a sequence all bins, except possibly two, have level at least $2/3$. First, we verify this for non-critical bins (bins without medium items). If a non-critical bin is opened by a large item, it clearly has level higher than $2/3$. All other non-critical bins only include items of size at most $1/2$. Hence, these bins, except possibly the last one, include at least two items. Among the non-critical bins that include two items, consider two bins $b_i$ and $b_j$ ($i<j$) that have levels smaller than $2/3$. Since $b_j$ contains at least two items, at least one of them has size smaller than $1/3$. This item could fit in $b_i$ by the \FF property. We conclude that all non-critical bins, except possibly two, have level at least $2/3$.
Now, suppose two critical bins $b_i$ and $b_j$ have levels smaller than $2/3$. Consider the first non-medium item $x$ which is packed in $b_j$ (in the second case, such an item exists). Since a medium item is packed in the bin, $x$ should be either tiny or small. If $x$ is small, then the level of $b_j$ is at least $1/2 + 1/3$, which contradicts the level of $b_j$ being smaller than $2/3$. Similarly, $x$ cannot be a tiny item of size larger than $1/6$ (since $1/2+1/6 \geq 2/3$). Hence, $x$ is a tiny item of size at most $1/6$. This implies that at the time the online algorithm packs $x$, bin $b_i$ has a virtual level of at least $5/6$. The virtual level is at most $1/6$ larger than the actual level (the final level). 
Hence, the actual level of $b_i$ is at least $5/6-1/6 = 2/3$. We conclude that at most one critical bin has level smaller than $2/3$. To summarize, at most three bins have level smaller than $2/3$. Hence, the cost of the algorithm is at most $3/2 \optsize + 3$. 
\qed
\end{proof}

\section{An Algorithm with Linear Advice}\label{linearAlg}

In this section, we present an algorithm that receives $2n + o(n)$ bits of advice and achieves a competitive ratio of $4/3+\varepsilon$ for any sequence of size $n$, and arbitrarily small (but constant) values of $\varepsilon$.
Consider an algorithm that receives an \textit{approximate size} for each sufficiently large item $x$ encoded using $k$ bits. The approximate size of $x$ would be larger than its \textit{actual size} by at most an additive term of $1/{2^k}$. The algorithm can optimally pack items by their approximate sizes and achieve an \textit{approximate packing} which includes a reserved space of size $x+\varepsilon$ $(\varepsilon \leq 1/2^k)$ for each item. Precisely, for each sufficiently large item $x$, the approximate packing includes a reserved space of size $x+\varepsilon$ $(\varepsilon \leq 1/2^k)$ for $x$. This enables the algorithm to place $x$ in the reserved space for it in the approximate packing. Smaller items are treated differently and the algorithm does not reserve any space for them. In the reminder of this section, we elaborate this idea to achieve a 4/3-competitive algorithm.

Notice that the cost of an approximate packing can be as large as $\frac{3}{2}$ times the cost of \opt. To see that, consider a sequence which is a permutation of $\langle \frac{1}{2}+ \varepsilon_1, \frac{1}{2} - \varepsilon_1, \frac{1}{2} + \varepsilon_2, \frac{1}{2} - \varepsilon_2 , \ldots, \frac{1}{2} + \varepsilon_{n/2}, \frac{1}{2}- \varepsilon_{n/2}\rangle $, where $\varepsilon_i < 1/2^n (1\leq i \leq n/2)$. Since \opt packs all bins tightly, an increase in the sizes of items by a constant (small) $\varepsilon$ results in opening a new bin for each two bins \opt uses. Hence the cost of the optimal approximate packing can be as bad as $\frac{3}{2}$ \opt.
This example suggests that using approximate packings is not good for the bins in which a small number of large items are tightly packed. To address this issue we divide the bins of \opt into two groups:

\begin{definition}\label{def1}
Consider an optimal packing of a sequence $\sigma$. Given a small parameter $\varepsilon' < 1/60$, define {\em good bins} to be those where the total size of the items smaller than $1/4$ in the bin is at least $5 \varepsilon'$.
Define all other bins to be {\em bad bins}.
\end{definition}
A part of the advice received for each item $x$ indicates if $x$ is packed by \opt in a good bin or in a bad bin. This enables us to treat items packed in these two groups separately.

\begin{lemma}\label{lemlin1}
Consider sequences for which all bins in the optimal packing are good (as defined above). There is an online algorithm that receives $o(n)$ bits of advice and achieves a competitive ratio of $4/3$.
\end{lemma}

\begin{proof}
Call an item \textit{small} if it is smaller than or equal to $1/6$ and \textit{large} otherwise. The advice bits define the approximate sizes of all large items with a precision of $\varepsilon'$. The amount of advice will be roughly $2^{1/ \varepsilon'} \log n$ which is $o(n)$ for constant values of $\varepsilon'$.
The online algorithm \alg can build the optimal approximate packing of large items. In such a packing, there is a reserved space of size at most $x+\varepsilon'$ for any large item of size $x$. The algorithm considers this packing as a partial packing and initializes the level of each bin to be the total sizes of approximated items in that bin. For packing an item $x$, if $x$ is large, \alg packs it in the space reserved for it in the approximate packing. It also updates the level of the bin to reflect the actual size of $x$. If $x$ is small, \alg simply applies the First-Fit strategy to pack $x$ in a bin of the partial packing (and opens a new bin for it if necessary). We prove that \alg is $4/3$-competitive. In the final packing by \alg, call a bin ``red'' if all items packed in it are small items and call it ``blue'' otherwise (the blue bins constitute the approximated packing at the beginning). There are two cases
to consider.

In the first case, there is no red bin in the final packing of \alg, i.e., all small items fit in the remaining space of the bins in the approximate packing of large items. Let $\sigma'$ be a copy of the input sequence in which the sizes of large items are approximated, i.e., increased by at most $\varepsilon'$; also let $X$ be the number of bins for the optimal packing of $\sigma'$. Since there is no red bin in the final packing of \alg, the cost of \alg is equal to $X$. Consider the optimal packing of the actual input sequence $\sigma$. Since all bins are good, one can transfer a subset of items to provide an available space of size at least $5\varepsilon'$ in each bin. After such a transfer, we can increase the sizes of large items to their approximate sizes. Since there are at most 5 large items in each bin and also available space of size at least $5\varepsilon'$, the packing constructed this way is a valid packing for the sequence $\sigma'$. Since the size of the transferred items for each bin is at most $1/4$, the transferred items from each group of four bins can fit in one new bin. Consequently the number of bins in the new packing is at most $5/4 \optsize$. We know that the final packing by \alg is the optimal packing for $\sigma'$ (with cost $X$), and in particular not worse than the packing constructed above. Hence, the cost of \alg is not more than $5/4 \optsize$.

In the second case, there is at least one red bin in the final packing of \alg. We claim that all bins in the final packing of \alg, except possibly the last, have levels larger than $3/4$. The claim obviously holds for the red bins since the levels of all these bins (excluding the last one) are larger than $5/6$. Moreover, since there is a bin which is opened by a small item, all blue bins have levels larger than $5/6$, i.e., the total size of packed items and reserved space for the large items is larger than $5/6$. Since there are at most 5 large items in each bin, the actual level of each bin in the final packing of \alg is at least $5/6 - 5\varepsilon'$, which is not smaller than $3/4$ for $\varepsilon' \leq 1/60$. So, all bins, except possibly one, have levels larger than $3/4$. Consequently, the algorithm is $4/3$-competitive. 
\qed
\end{proof}

It remains to address how to deal with bad bins. The next three lemmas do this.
\begin{lemma}\label{str2}
Consider sequences for which all bins in the optimal packing include precisely two items. There is an algorithm that receives $1$ bit of advice per request and achieves an optimal packing.
\end{lemma}

\begin{proof}
The single bit of advice for an item $x$ determines whether or not the \textit{p
artner} of $x$ appeared as a previous request, where the partner of $x$ is the i
tem which is packed in the same bin as $x$ in \opt's packing. Consider an algori
thm \alg that works as follows: If the partner of $x$ has not been requested yet
, \alg opens a new bin for $x$. Otherwise, it uses the BF strategy to pack $x$ i
n one of the open bins.%
We claim that \alg achieves an optimal packing.

Assume that initially we have a mapping that maps the last item to go into a bin
 to the item it goes on top of in the optimal packing, i.e., it maps the second 
item of each bin to the first item.
We update this mapping when necessary and maintain the invariant that
we can always pack optimally according to the mapping. For serving a request $x$
, if BF does
not pack according to this mapping, it packs $x$ on top of $y'$, while,
according to the mapping, it was supposed to pack $x$ on top of $y$,
and a later $x'$ is supposed to go on top of $y'$. Due to the BF strategy, $y' \
geq y$, so
we can update the mapping to map the currently unprocessed $x'$ to $y$,
and, of course, $x$ to $y'$.
\qed
\end{proof}

\begin{lemma} \label{str3}
Consider a sequence $\sigma$ for which all items have sizes larger than $1/4$ and for which each bin in \opt's packing includes precisely three items. The cost of the Harmonic algorithm is at most $4/3 \optsize + 3$ for serving such a sequence.
\end{lemma}

\begin{proof}
The proof is based on a simple weighting function. Call an item $x$ {\em large} if $1/3 < x <1/2$ and {\em small} otherwise ($1/4 < x \leq 1/3$). Define the weight of $x$ to be $1/2$ if $x$ is large and $1/3$ if it is small. Consider a bin $B$ in the packing of $\sigma$ by \opt. Since there are three items in $B$, its weight is maximized when there are two large items and one small item in it (three large item do not fit in the same bin). Hence, the weight of each bin in the \opt packing is at most $2 \times 1/2 + 1/3 = 4/3$. Consequently, we have $\optsize \geq 3/4 W$, where $W$ is the total weights of all items.

The Harmonic algorithm (\HA) simply packs small and large items in separate collections of bins. So, each of the algorithm's bins, except possibly two bins, contains either three small items or two large items. In both cases, the weight of each bin is at least $1$ and we have $\HA(\sigma) \leq W +2$. As a conclusion $\HA(\sigma) \leq 4/3 \opt(\sigma) + 2$ which completes the proof. 
\qed
\end{proof}

\begin{lemma}\label{lem4}
Consider a sequence $\sigma$ for which all bins in the optimal packing are bad bins (as defined earlier). There is an algorithm that receives two bits of advice for each request, and
opens at most $(4/3 + \frac{5 \varepsilon'}{1- 5\varepsilon'})\optsize+3$ bins.
\end{lemma}

\begin{proof}
By the definition of bad bins, for any bin in the optimal packing,
all items are either smaller than $5 \varepsilon'$ or larger than $1/4$. We call the former group of items \textit{tiny} items and pack them separately using the \FF strategy. We refer to other items as \textit{normal items}. Consider an offline packing \optp which is the same as \opt's packing, except that all tiny items are removed from their bins and packed separately in new bins using the \FF strategy. This implies that the cost of \optp is larger than $\optsize$ by a multiplicative factor of at most $1+ \frac{5 \varepsilon'}{1-5 \varepsilon'}$. Let \pdd be the optimal packing for normal items.
Since all normal items are larger than $1/4$, each bin of \pdd contains at most three items. \cmnt{Note that we compare our algorithm with the optimal algorithm for normal items. P is not necessarily optimal for these items.} We say a bin of \pdd has type $i$ $(i \in \SET{1,2,3} $), if it contains $i$ normal items. Similarly, we say an item $x$ has type $i$ if it is packed in a type $i$ bin. All items in type 3 bins have sizes smaller than $1/2$ (otherwise one will have size at most $1/4$ which contradicts the assumption). Moreover, the sizes of the items in all type $1$ bins (except possibly the last one) are larger than $1/2$ (otherwise a better packing is achieved by pairing two of them). With two bits of advice, we can detect the type of an item as follows:
Let $b$ denote the two bits of advice with item $x$. If $b$ is ``$01$'' and $x>1/2$, then $x$ has type $1$; if $b$ is ``$01$'' and $x\leq 1/2$, then $x$ has type $3$; and if $b$ is ``$10$'' or $b$ is ``$11$'', then $x$ has type $2$. Note that the code ``$00$'' is not used at this point (this is used later on),
and the use of ``$10$'' and ``$11$'' is still to be detailed.

Let $X_i$ denote the number of bins of type $i$ ($1\leq i \leq 3)$. Hence, the cost of \pdd is $X_1 + X_2 + X_3$, and consequently the cost of $\optp$ is at least $ X_1 + X_2 + X_3 + X'$, where $X'$ is the number of bins filled by tiny items.
Consider an algorithm \alg that performs as follows. If an item $x$ has type $1$, \alg simply opens a new bin for $x$. If $x$ has type $2$, \alg applies the strategy of Lemma~\ref{str2} to place it in one of the bins maintained for items of type $2$. Recall that the advice in this case is either ``$10$'' or ``$11$'', so the second bit provides the advice required by Lemma~\ref{str2}. If $x$ has type $3$, \alg applies the Harmonic strategy to pack the item in a set of bins maintained for type $3$ items. By Lemma~\ref{str3}, the cost of \alg for these items is at most $4/3 X_3 +3$. Finally, \alg uses the \FF strategy to pack tiny items in separate bins. Consequently, the cost of the algorithm is at most $X_1+X_2+4/3X_3 + X'+3 \leq (1+ \frac{5 \varepsilon'}{1- 5\varepsilon'}) \optsize + X_3/3 + 3 \leq (4/3+ \frac{5 \varepsilon'}{1- 5\varepsilon'}) \optsize +3$. 
\qed
\end{proof}

Provided with the above results, we arrive at the following result:

\begin{theorem}
There is an online algorithm which receives two bits of advice per request,
plus an additive lower order term, and achieves a competitive ratio of $4/3 + \varepsilon$, for any positive value of $\varepsilon$.
\end{theorem}

\begin{proof}
Define $\varepsilon'$ to be $\frac{11 \varepsilon}{60}$. For $\varepsilon < 1/11$, we have $\varepsilon' < 1/60$. Moreover, we have $\frac{5 \varepsilon'}{1-5 \varepsilon'} \leq \frac{5 \varepsilon'}{1-1/12} = \frac{60 \varepsilon'}{11} = \varepsilon$. In an optimal packing, divide bins into good and bad bins using Definition~\ref{def1}.
 Also, let $Gd$ and $Bd$ respectively denote the number of good and bad bins. Use advice bits to distinguish items which are packed in good and bad bins, and pack them in separate lists of bins. More precisely, let the two bits of advice for an item $x$ be ``$00$'' if it is packed by \opt in a good bin, and apply Lemma~\ref{lemlin1} to pack these items in at most $4/3 Gd$ bins. Similarly, apply Lemma~\ref{lem4} to pack items from bad bins in at most $(4/3 + \frac{5 \varepsilon'}{1-5 \varepsilon'}) Bd +3 \leq (4/3 + \varepsilon) Bd +3$ bins, using bits of advice of the form ``$01$'', ``$10$'', or ``$11$'', as discussed in the proof of Lemma~\ref{lem4}. Consequently, the cost of the algorithm will be at most $4/3 Gd + (4/3 +\varepsilon) Bd +3 \leq (4/3 +\varepsilon) \optsize +3$. 
\qed
\end{proof}

\section{A Lower Bound for Linear Advice}
The \textit{GMP problem}~\cite{EmekFraKorRos2011} and the \textit{String Guessing Problem}~\cite{BocHroKom13} both contain a core special case of guessing a binary sequence.
We use their results to show that an online algorithm needs a linear number of bits of advice to achieve a competitive ratio better than $9/8$ for bin packing.

\begin{definition}[\cite{EmekFraKorRos2011,BocHroKom13}]
The {\em Binary String Guessing Problem with known history ($2$-SGKH)} is the following online problem. The input $I = (n, \sigma = \left\langle x_1, x_2, \ldots, x_n \right\rangle)$ consists of $n$ items that are either ``$0$'' or ``$1$'' and that are revealed one by one. For each item $x_t$, the online algorithm \alg must guess if it is a ``$0$'' or a ``$1$''. After the algorithm has made a guess, the value of $x_t$ is revealed to the algorithm.
\end{definition}

\begin{lemma}[\cite{BocHroKom13}] \label{servi}
On any input of length $n$, any deterministic algorithm for $2$-SGKH that is guaranteed to guess correctly on more than $\alpha n$ bits, for $1/2 \leq \alpha < 1$, needs to read at least $(1 + (1 - \alpha) \log(1 - \alpha) + \alpha \log \alpha) n$ bits of advice.
\end{lemma}

Since the number of bits needed to express the number of ``0''s in the input is at most $\lceil \log (n+1) \rceil \leq \log n + 1$, and this number can be given as advice by an oracle, if it is not given to the algorithm otherwise, we easily obtain the following lemma.
Recall that the definition of $e$, the length of the encoding function, is given
in Section~\ref{the-model}.

\begin{lemma} \label{servi2}
Consider instances of size $n$ of the $2$-SGKH problem in which the number of ``0''s is given to the algorithm as part of the input. For these instances, any deterministic algorithm that is guaranteed to guess correctly on more than $\alpha n$ bits, for $1/2 \leq \alpha < 1$, needs to read at least $(1 + (1 - \alpha) \log(1 - \alpha) + \alpha \log \alpha) n - e(n)$ bits of advice.
\end{lemma}

\begin{proof}
Assume to the contrary that the statement is not true. Hence, there is an algorithm, \BSGA, that knows the number of ``0''s and receives fewer than $(1 + (1 - \alpha) \log(1 - \alpha) + \alpha \log \alpha) n - e(n)$ bits of advice while guessing correctly on more than $\alpha n$ bits. This algorithm can be used to serve arbitrary instances of the $2$-SGKH problem (in which the number of ``0''s is not known). Modify the advice tape used by the algorithm \BSGA so that it contains at most $e(n)$ additional bits at the beginning specifying the number of ``0''s.
(This can be done with the self-delimited encoding of the number of ``0''s.)
The algorithm for $2$-SGKH reads this number and gives it to \BSGA. Then it asks \BSGA for its guess for each bit in the sequence and answers the same as \BSGA. It also informs \BSGA of when it is correct and when it is wrong, with the same information it is given. The algorithm is correct exactly when \BSGA is correct.
The total number of advice bits will be less than $e(n) + (1 + (1 - \alpha) \log(1 - \alpha) + \alpha \log \alpha) n - e(n) = (1 + (1 - \alpha) \log(1 - \alpha) + \alpha \log \alpha) n$. However, Lemma~\ref{servi} implies that no algorithm can guess correctly on more than $\alpha n$ bits with this many bits of advice. In conclusion, the initial assumption
is incorrect and the statement holds. 
\qed
\end{proof}

In order to relate the Binary String Guessing Problem to the online bin packing problem, we introduce another problem called the Binary Separation Problem.

\begin{definition}
The {\em Binary Separation Problem} is the following online problem. The input $I = (n_1, \sigma= \left\langle y_1, y_2, \ldots, y_n \right\rangle)$ consists of $n = n_1+n_2$ positive values which are revealed one by one. There is a fixed partitioning of the set of items into a subset of $n_1$ {\em large} items and a subset of $n_2$ {\em small} items, so that all large items are larger than all small items. Upon receiving an item $y_i$, an online algorithm for the problem must guess if $y$ belongs to the set of small or large items. After the algorithm has made a guess, it is revealed to the algorithm whether $y_i$ actually belongs to class of small or large items.
\end{definition}

We provide reductions from the modified Binary String Guessing Problem to the Binary Separation Problem, and from the Binary Separation Problem to the online bin packing problem. In order to reduce a problem $P_1$ to another problem $P_2$, given an instance of $P_1$ defined by a sequence $\sigma_1$ and a set of parameters $\eta_1$ (such as the length of $\sigma_1$ or the number of ``0''s in it), we create an instance of $P_2$ which is defined by a sequence $\sigma_2$ and also a set of parameters $\eta_2$. In our reductions, we assume $\eta_2$ is derived from $\eta_1$, and since $\sigma_1$ is revealed in an online manner, $\sigma_2$ is created in an online manner by looking only at $\eta_1$ and the revealed items of $\sigma_1$.

\begin{lemma}\label{leman1}
Assume that there is an online algorithm that solves the Binary Separation Problem on sequences of length~$n$ with $b(n)$ bits of advice, and makes at most $r(n)$ mistakes. Then there is also an algorithm that solves the Binary String Guessing Problem on sequences of length $n$, assuming the number of ``0''s is given as a part of input, so that the algorithm receives $b(n)$ bits of advice and makes at most $r(n)$ errors.
\end{lemma}

\begin{algorithm}[!t]
\begin{algorithmic}[1]
\item[]{The Binary Guessing algorithm knows the number of ``0''s ($n_1$) and passes it as a parameter (the number of large items) to the Binary Separation algorithm}
\STATE small = 0; large = 1 
\REPEAT
  \STATE mid = (large $-$ small) / 2
   \STATE class\_guess = SeparationAlgorithm.ClassifyThis(mid) \label{separationGuess}
   \IF{class\_guess = ``large''}
      \STATE bit\_guess = 0
   \ELSE
      \STATE bit\_guess = 1
   \ENDIF
   \STATE actual\_bit = Guess(bit\_guess) \COMMENT{The actual value is received after guessing ($2$-SGKH).}
   \IF{actual\_bit = 0}
      \STATE large = mid \COMMENT{We let ``large'' be the correct decision.}
   \ELSE
      \STATE small = mid \COMMENT{We let ``small'' be the correct decision.}
   \ENDIF
\UNTIL{end of sequence}
\end{algorithmic}
\caption{Implementing Binary String Guessing via Binary Separation.}
\label{algorithm-reduction-one}
\end{algorithm}

\begin{proof}
We assume that we have an algorithm \BSA that solves the Binary Separation Problem under the conditions of the lemma statement. Using that algorithm, we define the number $n_1$ of large items to be the number of ``0''s in the instance of the Binary String Guessing Problem. Then, we implement our algorithm \BSGA for the Binary String Guessing Problem as outlined in Algorithm~\ref{algorithm-reduction-one}, which defines the reduction. 
This \BSGA implementation, defined in Algorithm~\ref{algorithm-reduction-one}, functions as an adversary for \BSA, e.g., in Line 4, \BSGA gives \BSA its next request. Notice that we ensure that the \BSGA makes a correct guess if and only if \BSA makes a correct guess. The advice tape is filled with
bits of advice for this combined algorithm. The \BSGA uses the \BSA as
a sub-routine, but all the questions are effectively coming from the \BSA.

The set-up, reminiscent of binary search, is carried out as specified in the algorithm with the purpose of ensuring that when the \BSA is informed of the actual class of the item it considered, no result can contradict information already obtained. Specifically, the next item for the \BSA to consider is always in between the
largest item which has previously been deemed ``small'' and the smallest item which has previously been deemed ``large''. The fact that we give the middle item from that interval is unimportant; any value chosen from the open interval would work.
\qed
\end{proof}

Now, we prove that if we can solve a special case of the bin packing problem, we can also solve the Binary Separation Problem.
\begin{lemma}\label{leman2}
Consider the bin packing problem on sequences of length~$2n$
for which \opt opens $n$ bins.
Assume that there is an online algorithm \A that solves the problem on these instances with $b(n)$ bits of advice and opens at most $n+r(n)/4$ bins. Then there is also an algorithm \BSA that solves the Binary Separation Problem on sequences of length~$n$ with $b(n)$ bits of advice and makes at most $r(n)$ errors.
\end{lemma}

\begin{proof}
In the reduction, we encode requests for the \BSA as items for bin packing.
Assume we are given an instance $I = (n_1, \sigma= \left\langle y_1, y_2, \ldots, y_n\right\rangle)$ of the Binary Separation problem, in which $n_1$ is the number of large items ($n_1+n_2 = n$), and the values of $y_t$s are revealed in an online manner $(1 \leq t \leq n)$. We create an instance of the bin packing problem which has length~$2n$. Algorithm~\ref{algorithm-reduction-two} shows the details of the reduction. The bin packing sequence starts with $n_1$ items of size $\frac{1}{2} + \varepsilon_{min}$ (in Algorithm~\ref{algorithm-reduction-two}, the variable ``NumberOfLargeItems'' is $n_1$ from the Binary Separation Problem). Any algorithm needs to open a bin for each of these $n_1$ items. We create the next $n$ items in an online manner, so that we can use the result of their packing to guess the requests for the Binary Separation Problem. Let $\tau = y_t$ $(1 \leq t \leq n)$ be a requested item of the Binary Separation Problem; we ask the bin packing algorithm to pack an item whose size is an increasing function of $\tau$, and slightly less than $\frac{1}{2}$. Depending on the decision of the bin packing algorithm for opening a new bin or placing the item in one of the existing bins, we decide the type of $\tau$ as being consecutively small or large. The last $n_2$ items of the bin packing instance are defined as complements of the items in the bin packing instance associated with small items in the binary separation instance (the complement of item $x$ is $1-x$). We do not need
to give the last items complementing the small items in order to implement the
algorithm, but we need them for the proof of the quality of the correspondence
that we are proving. 

Call an item in the bin packing sequence ``large'' if it is associated with large items in the Binary Separation Problem, and ``small'' otherwise.
For the bin packing sequence produced by the reduction, an optimal algorithm pairs each of the large items with one of the first $n_1$ items (those with size $\frac{1}{2} + \varepsilon_{min}$), placing them in the first $n_1$ bins. \opt pairs the small items with their complements, starting one of the next $n_2$ bins with each of these small items. Hence, the cost of an optimal algorithm is $n_1+n_2=n$.
The values $\varepsilon_{min}$ and $\varepsilon_{max}$ in Algorithm~\ref{algorithm-reduction-two} must be small enough so that
no more than two of any of the items given in the algorithm can fit together in a bin. No other restriction is necessary.

We claim that each extra bin used by the bin packing algorithm, but not by \opt, results in at most four mistakes made by the derived algorithm on the given instance of the Binary Separation Problem.
Consider an extra bin in the final packing of \A. This bin is opened by a large item which is incorrectly guessed as being small (bins which are opened by small items also appear in \opt's packing).  Note that large items do not fit in the same bins as complements of small items. The extra bin has enough space for another large item. Moreover, there are at most two small items which are incorrectly guessed as being large and placed in the space dedicated to the large items of the extra bin. Hence, there is an overhead of at least one for four mistakes. To summarize, \A has to decide if a given item is small or large and performs accordingly, and it pays a cost of at least $1/4$ for each incorrect decision. If \A opens at most $n + r(n)/4$ bins, the algorithm derived from \A for the Binary Separation Problem makes at most $r(n)$ mistakes.
\qed
\end{proof}

\begin{algorithm}[!t]
\begin{algorithmic}[1]
\STATE Choose $\varepsilon_{min}$ and $\varepsilon_{max}$ so that
 $0<\varepsilon_{min}<\varepsilon_{max}<\frac{1}{6}$
\STATE Choose a decreasing function
 $f\WEHAVE\REALS\rightarrow (\varepsilon_{min}..\varepsilon_{max})$
\FOR{i = 1 to NumberOfLargeItems}
 \STATE BinPacking.Treat($\frac{1}{2}+\varepsilon_{min}$) \COMMENT{The decision can only be to open a bin.}
\ENDFOR
\REPEAT
 \STATE Let $\tau$ be the next request
 \STATE decision = BinPacking.Treat($\frac{1}{2}-f(\tau)$)
 \IF{decision = ``packed with an $\frac{1}{2}+\varepsilon_{min}$ item''}
\STATE class\_guess = ``large''
 \ELSE
\STATE class\_guess = ``small''
 \ENDIF
 actual\_class = Guess(class\_guess)
 \IF{actual\_class = ``small''}
\STATE SmallItems.append($\frac{1}{2}-f(\tau)$) \COMMENT{Collecting small items for later.}
 \ENDIF
\UNTIL{end of request sequence}
\FOR{i = 1 to len(SmallItems)}
 \STATE BinPacking.Treat(1 $-$ SmallItems[i]) \COMMENT{The decision is not used.}
\ENDFOR
\end{algorithmic}
\caption{Implementing Binary Separation via Special Case Bin Packing.}
\label{algorithm-reduction-two}
\end{algorithm}


\begin{theorem}\label{sadt}
Consider the online bin packing problem on sequences of length~$n$. To achieve a competitive ratio of $c$ $(1 < c < 9/8)$, an online algorithm needs to receive at least $(n (1+ (4c -4) \log (4c -4) + (5-4c) \log (5-4c)) - (\lceil \log (n+1) \rceil+ 2\lceil\log (\lceil\log (n+1) \rceil+1) \rceil + 1))/2$ bits of advice.
\end{theorem}

\begin{proof}
Consider a bin packing algorithm \A that receives $b(n)$ bits of advice and achieves a competitive ratio of $c$. This algorithm opens at most $(c-1) \optsize$ bins more than \opt, so when $\optsize = n/2$, it opens at most $(c-1)n/2$ more bins. By Lemma~\ref{leman2}, the existence of such an algorithm implies that there is an algorithm \A that solves the Binary Separation Problem on sequences of length~$n/2$ with $b$ bits of advice and makes at most $2(c-1)n$ errors. By Lemma~\ref{leman1}, this implies that there is an algorithm \B that solves the Binary String Guessing Problem on sequences of length~$n/2$ with $b$ bits of advice and makes at most $2(c-1)n$ mistakes, i.e., it correctly guesses the other $n/2 - 2(c-1)n = (5-4c) n/2$ items. Let $\alpha = 5-4c$, and note that $\alpha$ is in the range $[1/2,1)$ when $c$ is in the range $(1,9/8]$. Lemma~\ref{servi2} implies that in order to correctly guess more than $\alpha n/2$ of the items in the binary sequence, we must have $b(n)$ larger than or equal to $((1 + (1 - \alpha) \log(1 - \alpha) + \alpha \log \alpha) n - e(n) )/2$. Replacing $\alpha$ with $5-4c$ completes the proof. 
\qed
\end{proof}

Thus, to obtain a competitive ratio strictly better than $9/8$, a linear number of bits of advice is required. For example, to achieve a competitive ratio of $17/16$, at least $0.188n$ bits of advice are required asymptotically.

\begin{corollary}
Consider the bin packing problem for packing sequences of length~$n$. To achieve a competitive ratio of $9/8-\delta$, in which $\delta$ is a small, but fixed positive number, an online algorithm needs to receive $\Omega(n)$ bits of advice.
\end{corollary}

\section{Concluding Remarks}
We conjecture that a sublinear number of bits of advice is enough to achieve competitive ratios smaller than $4/3$. Note that our results imply that we cannot hope for ratios smaller than $9/8$ with sublinear advice.







\bibliography{confshort,online}

\newpage
\appendix

\end{document}